\documentclass{llncs}

\usepackage{graphicx}
\usepackage{amsmath}
\usepackage{amssymb}
\usepackage{tabularx}
\usepackage{booktabs}
\usepackage{csquotes}
\usepackage[sort]{cite}
\usepackage[hyphens]{url}
\usepackage{hyperref}
\usepackage{pgfplots}
\pgfplotsset{compat=1.5}

\usepackage{algorithmic}
\usepackage{algorithm}

\floatname{algorithm}{Procedure}

\usepackage{tikz}
\usetikzlibrary{positioning}
\usetikzlibrary{decorations.pathreplacing}
\usetikzlibrary{decorations.text}
\usetikzlibrary{calc}
\usetikzlibrary{shapes}

\usepackage{xspace}
\newcommand{\ie}{i.\,e.\xspace}

\newcommand{\eg}{e.\,g.\xspace}

\begin{document}
\frontmatter          %
\pagestyle{headings}  %
\mainmatter           %
\title{Towards a Concurrent and Distributed Route Selection for Payment Channel Networks}
\titlerunning{TBA}
\author{Elias Rohrer\inst{1} \and Jann-Frederik La\ss{}\inst{2} \and Florian Tschorsch\inst{1}}
\authorrunning{Rohrer et al.} %
\institute{Technical University of Berlin,\\
\email{\{elias.rohrer, florian.tschorsch\}@tu-berlin.de}
\and
Humboldt University of Berlin,\\
\email{lassjann@informatik.hu-berlin.de}}

\maketitle

\begin{abstract}
\vspace{-2ex}
Payment channel networks use off-chain transactions to provide virtually
arbitrary transaction rates.
In this paper, we provide a new perspective on payment channels and consider them as a flow network.
We propose an extended push-relabel algorithm to find payment flows in a payment channel network.
Our algorithm enables a distributed and concurrent execution without violating capacity constraints.
To this end, we introduce the concept of capacity locking. %
We prove that flows are valid and present first results.
\vspace{-2ex}
\end{abstract}

\section{Introduction}

It seems that blockchain-based systems such as Bitcoin~\cite{nakamoto2008bitcoin} will,
due to their requirements regarding storage, processing power, and bandwidth,
not be able to natively scale to high transaction rates~\cite{croman2016scalingblockchains}.
Off-chain approaches~\cite{wattenhofer2015duplexmicropayment,poon2015bitcoin}, however,
offer a way to create long-lived payment channels between two nodes.
The payments transferred via a payment channel are processed locally
and therefore eliminate the need to commit each individual transaction to the blockchain.

In order to enable payments between any two nodes---whether they are directly connected or
not---payment channels form a network in which payments can be routed over more than one hop.
Finding a route that can process a certain transaction volume is challenging, though.
Related approaches~\cite{prihodko2016flare} cannot guarantee to utilize the available capacities
as they focus on finding a single path. %
We argue that single-path routing restricts the transferable amount
and misses many payment opportunities due to bottleneck capacities in the network.
Eventually, failed payments will fall back to on-chain transactions,
instead of using the available (and already locked) resources efficiently.

In this paper, we propose to aggregate multiple paths to a \emph{payment flow},
which can in sum provide larger transaction volumes.
We believe that algorithms from the domain of flow networks in general
and the push-relabel algorithm~\cite{goldberg1988newapproach} in particular are appropriate candidates
for route selection in payment networks.

Our main contribution is an algorithm for distributed route selection, which is based on the push-relabel algorithm.
It can find feasible flows in a payment channel network and is safe for
concurrent execution. To this end, we introduce the concept of \emph{capacity locking}.
We show that our algorithm guarantees that routes are feasible flows
and at the same time does not violate any capacity constraints.
Our first results confirm that the approach
is able to handle a high number of flows and transaction
volumes. The results emphasize that our approach succeeds in scenarios where
single-path routing schemes are bound to fail.
In summary, we offer a new perspective on payment channel networks.

The remainder is structured as follows.
Sec.~\ref{sec:relwork} discusses related work.
Subsequently, Sec.~\ref{sec:paymentflows} introduces payment flows and describes the basic algorithmic design.
Sec.~\ref{sec:distributedpaymentflows} develops a distributed and concurrent route selection
algorithm. In Sec.~\ref{sec:evaluation}, we present and discuss first results,
before Sec.~\ref{sec:conclusion} concludes the paper.

\section{Background and Related Work}
\label{sec:relwork}
Payment channels are a new and unexplored concept.
The specifications~\cite{lightning-rfc-git} of the Lightning Network~\cite{poon2015bitcoin},
for example, are subject to constant change.
For the sake of clarity, we abstract from any specific payment channel design~\cite{wattenhofer2015duplexmicropayment,poon2015bitcoin}.

Routing in a payment channel network poses many challenges,
\eg, regarding the routing paradigm (per-hop routing vs.\ source routing)
and the topology (hub-and-spoke vs.\ peer-to-peer).
In this paper, we focus on route selection,
\ie, finding a route in a payment channel network that meets certain constraints.
Flare~\cite{prihodko2016flare}, a routing system for the Lightning Network,
creates a list of candidate routes from the set of channels with sufficient capacity.
So far, however, Flare and the Lightning Network opt for single-path routes.
In our work, we consider a payment as a flow
and elaborate the possibility to aggregate multiple paths.

We identify flow network algorithms as a promising direction to find multi-path routes.
While multi-commodity flows address a similar problem,
most of the existing approaches require global knowledge and/or a centralized routing coordinator.
The approach in~\cite{awerbuch1994improved} allows a distributed and concurrent execution
but solves the feasible-flow problem only approximately.
Our distributed algorithm, in contrast, guarantees that the selected route is a feasible flow.
Moreover, it can be executed concurrently without violating capacity constraints.

\section{Payment Flows}
\label{sec:paymentflows}

Payment flows describe a flow of units between pairs of nodes in a payment channel network.
Figure~\ref{fig:example} shows an example of a payment channel network
in which node $s$ wants to send a payment to node $t$.
We consider the payment channel network as a peer-to-peer network
in which nodes communicate directly with each other
and build an overlay network congruent with the payment channel network.
That is, we aim for a decentralized route selection.

\begin{figure}[t]
	\begin{center}
		\begin{minipage}{.6\textwidth}
		\begin{tikzpicture}[node distance=0.4cm]
			\tikzstyle{vertices}=[draw,circle,minimum height=20pt,outer
			sep=5pt, thick]
			\tikzstyle{pill}=[draw,rectangle,rounded corners=10pt,minimum
			width=40pt, thick]
			\tikzstyle{edges}=[midway, ->, thick]
			\node [vertices,pill] (s)  at (0,0) {$v_1 = s$};
			\node [vertices, above right = 0.2cm and 1.5cm of s] (v2) {$v_2$};
			\node [vertices, below right = 0.2cm and 1.5cm of s] (v3) {$v_3$};
			\node [vertices,pill, above right = 0.2cm and 1.5cm of v3] (t) {$v_4 = t$};

			\draw [edges] (s) to node  [auto] {$3$} (v2);
			\draw [edges] (s) to node  [auto, swap] {$2$} (v3);
			\draw [edges] (v2) to node [auto] {$2$} (v3);
			\draw [edges] (v2) to node [auto] {$1$} (t);
			\draw [edges] (v3) to node [auto,swap] {$3$} (t);
		\end{tikzpicture}
	\end{minipage}
	\hfill
	\begin{minipage}{.30\textwidth}
		\begin{tabularx}{\textwidth}{Xr}
		\toprule
		\textbf{Path} & \textbf{Vol.} \\
		\midrule
		$s \rightarrow v_2 \rightarrow t$ & 1\\
		$s \rightarrow v_2 \rightarrow v_3 \rightarrow t$ & 2\\
		$s \rightarrow v_3 \rightarrow t$ & 2\\
		\midrule
		maximum flow & 4\\
		\bottomrule
		\end{tabularx}
	\end{minipage}
	\end{center}
	\vspace{-1ex}
	\caption{Payment channel network example.}
	\label{fig:example}
	\vspace{-1ex}
\end{figure}
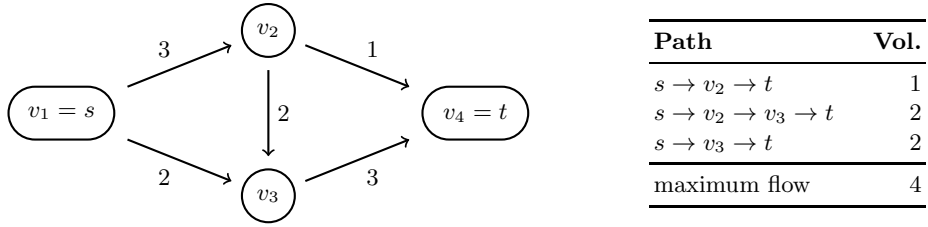

In order to process the payment, a path between $s$ and $t$ must exist.
Every path is a concatenation of payment channels.
Since payment channels have a capacity,
as indicated by the edge labeling in Figure~\ref{fig:example},
a path's transaction volume is limited
by the smallest payment channel capacity of this path.
While we cannot eliminate this limit, we can use multiple paths,
which in sum provide a higher transaction volume.

Determining the maximum transferable amount poses a challenge.
For example, simply finding all paths from source to sink and summing up their
respective capacities does not suffice;
paths may have common edges and thus need to share the respective capacities.
For the example in Figure~\ref{fig:example}, this naive approach would violate payment channel capacities.

The problem of finding the largest payment flow between two nodes $s$ and $t$
in a capacitated flow network is known as the \emph{maximum-flow problem}.
Several algorithmic solutions to the maximum-flow problem exist.
In the following, we elaborate on the efficient and well-studied
\emph{push-relabel}~\cite{goldberg1988newapproach} algorithm
and adopt it for the route selection of payment flows in payment channel networks.

We consider a network of payment channels as a directed graph $G = (V,E)$
and a non-negative function $c: V \times V \rightarrow \mathbb{R}_{\geq 0}$.
We call $c$ the capacity function, which determines a channel's
capacity $c(u, v)$ with $u,v \in V$ and $(u,v) \in E$.
Moreover, nodes $s$ and $t$ are the source and sink of the flow.
The resulting network $F = (G,c,s,t)$ is called a \emph{flow network}.

\begin{definition}[pseudo-flow, pre-flow, feasible flow]
	A \textbf{pseudo-flow} on the capacitated graph $(G, c)$ is a mapping $f: V \times V \to \mathbb{R}$ with the properties:
		\begin{align*}
		f(u, v) &\leq c(u, v),\, \forall (u, v) \in E && \text{(capacity constraint)}\\
		f(u, v) &= - f(v, u),\, \forall (u, v) \in E && \text{(skew symmetry)}
		\end{align*}

		Note that pseudo-flows do not require incoming and outgoing flows of a node to
		be equal. Therefore, nodes can hold \textbf{excess flow}, denoted by
		\begin{align*}
			x_f(u) = \sum\limits_{v \in V} f(v, u) - \sum\limits_{v \in V}
			f(u, v).
		\end{align*}
		A pre-flow and a feasible flow are special kinds of pseudo-flows with one of the following constraints.
		A \textbf{pre-flow} requires
		\begin{align*}
		x_f(v) \geq 0,\, \forall v \in V \setminus \{s, t\} && \text{(non-negativity constraint)}
		\end{align*}
		and a \textbf{feasible flow} requires
		\begin{align*}
		x_f(v) = 0,\, \forall v \in V \setminus \{s, t\} && \text{(conservation constraint).}
		\end{align*}

	\end{definition}
\begin{definition}[residual capacity and residual graph]
	The \textbf{residual capacity} $c_f$ with regard to the pseudo-flow $f$ of an edge
	$(u, v) \in E$ is defined as the difference between the edge's capacity
	and its flow:
	\begin{align*}
		c_f(u, v) = c(u, v) - f(u, v).
	\end{align*}

	Then, the \textbf{residual graph} $G_f(V, E_f)$ indicates when changes can
	be made to flow $f$ in the network $G(V, E)$,
	where
	\begin{align*}
		E_f = \{(u, v) \in V \times V : c_f(u, v) > 0\}.
	\end{align*}

	Note, that edges $(u, v)$ do not have to be in the original set of
	edges $E$.
\end{definition}

\begin{definition}[height function]
	A mapping $h: V \to \mathbb{N}$ is a \textbf{height function} for the
	push-relabel algorithm, if
	\begin{align*}
		h(s) = |V|, \qquad
		h(t) = 0, \qquad
		h(u) \leq h(v) + 1,\, \forall (u,v) \in E_f.
	\end{align*}
\end{definition}

At the beginning, the generic push-relabel algorithm initializes node heights and flow excess, as well as the edge pre-flow
values with $0$.
Please note that source node $s$, in contrast to all other nodes, is set to a height $|V|$.
Moreover, $s$'s outgoing edges are saturated according to the height function's third condition.
After these initialization steps, the algorithm
repeatedly selects a node $u$ as active node and applies one of the two basic
operations \texttt{push} and \texttt{relabel}.
Both operations have mutually exclusive conditions,
which ensure that either \texttt{push} or \texttt{relabel} is applicable at a time.

The \texttt{push} procedure (cf.\ Procedure~\ref{alg:push}) tries to push
an excess $\delta$ from node $u$ towards a neighbor $v$ with a smaller height.
The maximum possible $\delta$ is determined as
the minimum between the excess flow and the residual capacity of edge $(u,v)$.
Accordingly, edge capacities and excess values are updated to reflect flow changes in the residual graph.
The procedure requires that $u$ has excess flow
and that an unsaturated edge $(u,v)$ to a neighbor $v$ one level below $u$ exists.

Eventually, node $u$ will saturate all outgoing edges that lead to neighbors on a lower level.
In this case, the \texttt{relabel} procedure (cf.\ Procedure~\ref{alg:relabel})
\enquote{raises} node $u$ to a higher level.
The procedure calculates the minimal height of its neighbor nodes
and sets $u$'s height to the level above this minimum.
Therefore, the excess of node $u$ is guaranteed to be \enquote{pushable} in the next step.

\begin{figure}[t]
\vspace{-1em}
\begin{algorithm}[H]
	\caption{\texttt{push(u,v)}}
\label{alg:push}
\begin{algorithmic}%
	\ENSURE $x_f(u) > 0, c(u,v) > 0, h(u) = h(v) + 1$
    \STATE $\delta := \text{min}(\, x_f(u),\; c_f(u,v) \,)$ \label{alg:push:delta}
    \STATE $f(u,v) := f(u,v) + \delta$;\; $f(v,u) := f(v,u) - \delta$
    \STATE $x_f(u) := x_f(u) - \delta$;\; $x_f(v) := x_f(v) + \delta$
\end{algorithmic}
\end{algorithm}
\vspace{-3.6em}
\begin{algorithm}[H]
	\caption{\texttt{relabel(u)}}
\label{alg:relabel}
\begin{algorithmic}%
    \ENSURE $x_f(u) > 0, \forall (u, v) \in E: h(u) \leq h(v)$
	\STATE $h(u) := 1 + \text{min}\left(\, h(v) : (u, v) \in E \,\right)$
\end{algorithmic}
\end{algorithm}
\vspace{-3em}
\caption{Push-relabel algorithm~\cite{goldberg1988newapproach}, which solves the maximum-flow and the feasible-flow problem
in flow networks.}
\vspace{-1.5em}
\end{figure}

The generic push-relabel algorithms continues until the conditions fail for all nodes.
That means, the highest possible transaction volume has been pushed to the sink $t$
and all network excess has been pushed back to the source,
\ie, $x_f(v) =  0,\, \forall v \in V$.
At this point, the push-relabel algorithm has transformed the pre-flow into a maximum flow
and hence solved the maximum-flow problem.

In a payment channel network, however, it is often not necessary to know the maximum transaction volume.
Rather, we want to find a payment flow that can process a certain amount only.
This is a slightly different problem, which is known as the \emph{feasible-flow problem}.
Fortunately, the push-relabel can easily be modified to solve the feasible-flow problem:
in order to find a payment flow from source $s$ to sink $t$ with a transaction volume $d$,
we can simply insert a new (virtual) node to the payment network.
We call it the \emph{pre-source} $s'$, with a single edge $(s', s)$ and capacity $c(s', s) = d$.
The virtual edge caps the transferable amount at exactly $d$.
We can now still apply the push-relabel algorithm, as described before, to find a feasible-flow in this network.

So far, we assumed only one instance of the push-relabel algorithm.
If multiple flows ought to be found subsequently in the same network,
the initial flow of one instance is the result of the last instance.
A generalization for subsequent flows, however, is easily possible.
The following section is dedicated to show how the push-relabel algorithm can be adapted
to enable route selection for concurrent and distributed payment flows.

\section{Concurrent and Distributed Payment Flows}\label{sec:distributedpaymentflows}

In payment channel networks,
it is desirable to allow a concurrent execution of the route selection algorithm.
To this end, simply running multiple instances of the push-relabel algorithm in parallel is not enough:
one instance for flow $f_1$, for example, could consume the reverse edges' residual capacity
that belong to another instance for flow $f_2$.
We call this issue \emph{capacity stealing}.

The problem domain of finding flows $f_1, \dots, f_k$ for $k$ commodities
with source-sink pairs $(s_1, t_1), \dots, (s_k, t_k)$ that meet the total capacity constraint
\begin{align*}
	F(u,v) = \sum\limits_{i=1}^{k} f_i(u,v) \leq c(u,v),\, \forall (u,v) \in E,
\end{align*}
are known as \emph{multi-commodity flow problems}.

As our main contribution,
we propose a modified push-relabel algorithm that allows to find feasible
flows in a concurrent multi-commodity scenario. To this end, we introduce the
concept of \emph{capacity locking}:
flow volumes are accounted for every commodity independently, while still
respecting each payment channel's total capacity constraint.
The capacities on the reverse edges created by a flow $f_1$ are therefore \emph{locked} for another flow
$f_2$, which prevents capacity stealing.

\begin{definition}[locked capacities and new residual capacity]
\label{def:capacitylocking}
Let
the \textbf{locked capacity} and \textbf{total locked
capacity} of flow $f_i$ on edge $(u,v)$ be
\begin{align*}
	l_i(u,v) = \text{max}(0, f_i(u,v))\qquad \text{and} \qquad  L(u,v) =
	\sum\limits_{i=1}^{k} l_i(u,v).
\end{align*}

Accordingly, the \textbf{residual capacity} is redefined as
\begin{align*}
	c_i(u,v) &= c(u,v) - L(u,v) + l_i(v, u),
\end{align*}
which yields an individual residual graph $G_i(V, E_i$) for each commodity $i$.
\end{definition}

Definition~\ref{def:capacitylocking} ensures
that there is always enough residual capacity on the reverse edges available
to push the existing excess back to the source.
Except for this augmented definition of the residual capacity,
the \texttt{locked-push} procedure (cf.\ Procedure~\ref{alg:lockedpush})
is similar to the original \texttt{push} procedure.
Note, however, that the modified push-relabel algorithm does not necessarily yield optimal flows
in the multi-commodity scenario.
It guarantees \emph{validity}, though,
which makes it superior compared to other approaches from this domain~\cite{awerbuch1994improved}.

In the following, we prove validity for our proposed algorithm.
As the skew-symmetry and flow-conservation constraints follow directly from
the definition of the algorithm, it suffices to show that it yields flows that
respect the total capacity constraint.
\begin{lemma}
	The total capacity constraint $F(u,v) \leq c(u,v),\, \forall (u,v) \in E$
	is never violated.
\end{lemma}
\begin{proof}
	For a \texttt{locked-push} of commodity $i$ on edge $(u,v)$, the
	change in flow volume $\delta$ is always chosen to be at maximum the remaining residual capacity of
	the flow on this edge. Accordingly, lock $l_i(u,v)$ cannot be greater
	than $\delta$. Therefore, the locked capacity never exceeds the edge
	capacity for each individual edge. It follows that the total
	capacity constraint is never violated:
	\begin{align*}
		F(u,v) = \sum\limits_{i=1}^{k} f_i(u,v) \leq L(u,v) = \sum\limits_{i=1}^{k} l_i(u,v)
		\leq \sum\limits_{i=1}^{k} c_i(u,v) \leq c(u,v).
		\raisebox{-5mm}{\hspace{1em}$\square$}
	\end{align*}
\end{proof}

\begin{figure}[t]
\vspace{-1em}
\begin{algorithm}[H]
\caption{\texttt{locked-push(i,u,v)}}
\label{alg:lockedpush}
\begin{algorithmic}%
    \ENSURE $x_i(u) > 0, c_i(u,v) > 0, h_i(u) > h_i(v)$
    \STATE $l_i(u,v) := \text{max}(0,f_i(u,v))$; $l_i(v,u) := \text{max}(0,f_i(v,u))$
	\STATE $c_i(u, v) := c(u,v) - L(u,v) + l_i(v,u)$
	\STATE $\delta := \text{min}(x_i(u), c_i(u,v))$
	\STATE $f_i(u,v) := f_i(u,v) + \delta$; $f_i(v,u) := f_i(v,u) - \delta$
	\STATE $L(u,v) := L(u,v) + \delta$; $L(v,u) := L(v,u) - \delta$
    \STATE $x_i(u) := x_i(u) - \delta$;\; $x_i(v) := x_i(v) + \delta$
\end{algorithmic}
\end{algorithm}
\vspace{-3em}
\caption{Capacity locking enables concurrent push-relabel execution without violating capacity constraints, \ie, capacity stealing.}
\vspace{-1.5em}
\end{figure}

In order to execute the modified algorithm in a distributed scenario,
the asynchronous distributed algorithm,
introduced in~\cite{goldberg1988newapproach}, is adapted to our needs:
each node maintains a local view on flow states, channel capacities, and its neighbors' height.
Furthermore, each node maintains routing information and its own height.
Then, every node $u$ with positive excess tries to push its excess
along an unsaturated outgoing edge to a neighbor $v$ of smaller height.
A \texttt{locked-push} can only be committed,
if $v$ acknowledges $u$ that it is has indeed a smaller height.
Alternatively, $v$ can reject the \texttt{locked-push} and respond with its actual height.
This way, $u$ learns its neighbors' height and can trigger \texttt{relabel}, if necessary.
After relabeling, $u$ sends height updates to its neighbors. The source and
sink node can determine the termination of the algorithm and communicate the result to
finalize route selection.

\section{Evaluation}\label{sec:evaluation}
In order to evaluate our approach,
we constructed a Watts-Strogatz graph with $\beta = 0.5$, $n = 200$, and a node degree of $10$.
Channel capacities were generated by uniform random sampling from $[0, 10]$.
In the following, we compare the sequential (seq.) and the concurrent (conc.) algorithm.

First, we are interested in the number of flows that each algorithm can handle.
To this end, we sampled the transaction volume from $[0, 20]$ and calculated
the mean success rate over $10$ runs, \ie, the share of successfully found flows.
The results, shown on the left of Figure~\ref{fig:success_flows},
indicate that both algorithms are able to find a large number of flows (relative to the network size).
At some point, when network capacities are exhausted, the success rate eventually drops.
Single-path approaches, in contrast, achieve in the best case a $0.5$ success
rate (cf.\ horizontal line in the plot):
while the maximum channel capacity is $10$,
on average every second transaction volume is in $(10,20]$
and therefore not feasible with a single path.
Effectively, this reduces the utilization of the available capacities by 50\%.

Second, we are interested in the transaction volume that we can achieve by aggregating multiple paths.
To this end, we set the number of flows to $128$, increased the transaction
volume, and calculated the mean success rate.
The results, shown on the right of Figure~\ref{fig:success_flows},
suggest that again both variations are able to route relatively large volumes.
In more than 50\% of the cases, the concurrent algorithm still manages to
process all $128$ flows for up to a volume of $15$ each.
This is especially noteworthy,
as a single-path approach would not be able to route a single payment with a volume exceeding $10$
in our scenario (cf.\ vertical line in the plot).
These first results illustrate that our approach is superior compared to
single-path route selection schemes.

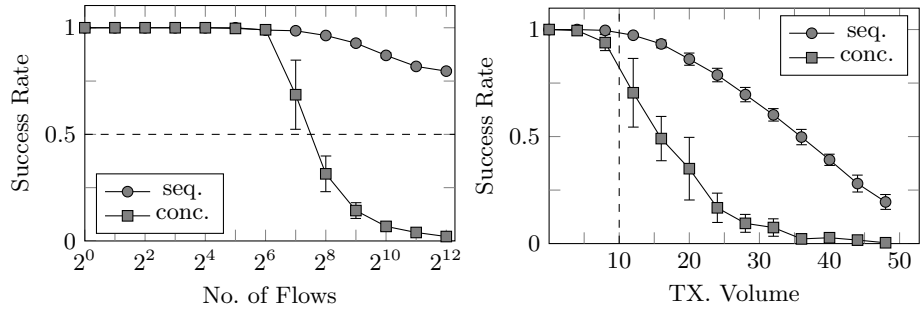
\begin{figure}[t]
	\begin{minipage}{.5\textwidth}
		\begin{tikzpicture}
			\begin{axis}[
				width=20em,
				height=14.5em,
				xlabel={No. of Flows},
				ylabel={Success Rate},
				xmode=log,
				log basis x={2},
				xtick={1,4,16,64,256,1024,4096},
				ytick={0,0.5,1.0},
				label style={font=\small},
				tick label style={font=\small},
				ymin=0,
				xmin=1,
				xmax=5000,
				extra x ticks={2,8,32,128,512,2048},
				extra x tick labels={},
				extra y ticks={0.25,0.75},
				extra y tick labels={},
				legend entries={seq.,conc.},
				legend pos=south west,
				cycle list name=black white,
				]
				\addplot+[error bars/.cd, y dir = both, y explicit]
				table[
				x=r,
				y=sbs_suc,
				y error=ci] 
				{data_flow_aggregated_sbs.dat};
				\addplot+[error bars/.cd, y dir = both, y explicit]
				table[
				x=r,
				y=cnc_suc,
				y error=ci] 
				{data_flow_aggregated_cnc.dat};
				\addplot[dashed, mark=none] coordinates {(1,0.5) (5000,0.5)};
			\end{axis}
		\end{tikzpicture}
		\end{minipage}
	\begin{minipage}{.5\textwidth}
		\begin{tikzpicture}
			\begin{axis}[
				width=20em,
				height=14.5em,
				xlabel={TX. Volume},
				ylabel={Success Rate},
				ytick={0,0.5,1.0},
				xtick={10, 20, 30, 40, 50},
				label style={font=\small},
				tick label style={font=\small},
				ymin=0,
				ymax=1.1,
				xmin=0,
				extra x ticks={5,15, 25, 35, 45},
				extra x tick labels={},
				extra y ticks={0.25, 0.75},
				extra y tick labels={},
				legend entries={seq.,conc.},
				legend pos=north east,
				cycle list name=black white,
				]
				\addplot+[error bars/.cd, y dir = both, y explicit]
				table[
				x=max_demand,
				y=sbs_suc,
				y error=ci] 
				{data_demand_aggregated_sbs.dat};
				\addplot+[error bars/.cd, y dir = both, y explicit]
				table[
				x=max_demand,
				y=cnc_suc,
				y error=ci] 
				{data_demand_aggregated_cnc.dat};
				\addplot[dashed, mark=none] coordinates {(10,0) (10,1.5)};
			\end{axis}

		\end{tikzpicture}
		\end{minipage}
	\caption{Flow Network Simulation: mean success rate over $10$ runs, dependant on the number of
	flows and transaction volume. Error bars show the 95\% confidence interval.}
	\label{fig:success_flows}
\end{figure}

\section{Conclusion}
\label{sec:conclusion}
In this paper, we argued that currently deployed single-path routing
schemes for payment channel networks suffer from a number of drawbacks.
Most prominently, they utilize the available capacities in the network inefficiently.
Eventually, single-path routes will lead to on-chain transactions as a fallback strategy
and therefore subvert the idea of payment channels.

We addressed this issue by presenting a novel perspective on route selection
that considers payment channel networks as flow networks.
Flow network algorithms utilize the available capacity by aggregating multiple paths,
which allow to route transactions of larger volume.
We proposed an extended push-relabel algorithm that finds flows based on
local knowledge. Thus, it is suitable for the concurrent and distributed
scenario encountered in payment channel networks.
We proved the validity of the flows and showed that our algorithm is indeed able to satisfy demands,
where single-path based approaches fail.

\end{document}